\documentclass[a4paper,11pt]{article}
\usepackage[margin=1in]{geometry}
\usepackage{amsfonts}
\usepackage{amsthm}
\usepackage{amsmath}
\usepackage{amssymb}
\usepackage{blkarray}
\usepackage[dvipsnames]{xcolor}
\usepackage{graphicx} 
\usepackage{subcaption}
\usepackage{mathtools}
\usepackage{underscore} 
\usepackage{nicematrix}

\usepackage{newtxtext,newtxmath}

\usepackage{enumitem}
\usepackage{hyperref}
\usepackage{crossreftools}
\usepackage{tikz}
\usepackage{tikzit}
\usetikzlibrary{quantikz2} 
\usepackage{algorithm}
\usepackage{algpseudocodex}
\usepackage{bm}
\usepackage{stmaryrd}
\usepackage{anyfontsize}
\usepackage{dsfont} 
\usepackage{multicol}
\usepackage{xspace}
\usepackage{lipsum}
\usepackage{booktabs}
\usepackage{authblk}
\usepackage{pifont}
\usepackage{aligned-overset}

\tikzstyle{box}=[shape=rectangle, text height=1.5ex, text depth=0.25ex, yshift=0.5mm, fill=white, draw=black, minimum height=5mm, yshift=-0.5mm, minimum width=5mm, font={\small}]
\tikzstyle{Z dot}=[inner sep=0mm, minimum size=2mm, shape=circle, draw=black, fill={rgb,255: red,221; green,255; blue,221}]
\tikzstyle{Z phase dot}=[minimum size=1.2em, font={\footnotesize\boldmath}, shape=rectangle, rounded corners=0.5em, inner sep=0.2em, outer sep=-0.2em, scale=0.8, tikzit shape=circle, draw=black, fill={rgb,255: red,221; green,255; blue,221}, tikzit draw=blue]
\tikzstyle{X dot}=[Z dot, shape=circle, draw=black, fill={rgb,255: red,255; green,136; blue,136}]
\tikzstyle{X phase dot}=[Z phase dot, tikzit shape=circle, tikzit draw=blue, fill={rgb,255: red,255; green,136; blue,136}, font={\footnotesize\boldmath}]
\tikzstyle{hadamard}=[fill=yellow, draw=black, shape=rectangle, inner sep=0.6mm, minimum height=1.5mm, minimum width=1.5mm]
\tikzstyle{vertex}=[inner sep=0mm, minimum size=1mm, shape=circle, draw=black, fill=black]
\tikzstyle{vertex set}=[inner sep=0mm, minimum size=1mm, shape=circle, draw=black, fill=white, font={\footnotesize\boldmath}]
\tikzstyle{target}=[inner sep=0mm, minimum size=3mm, shape=circle, draw=black]

\tikzstyle{hadamard edge}=[-, dashed, dash pattern=on 2pt off 1.5pt, thick, draw={rgb,255: red,68; green,136; blue,255}]
\tikzstyle{brace edge}=[-, tikzit draw=blue, decorate, decoration={brace,amplitude=1mm,raise=-1mm}]
\tikzstyle{diredge}=[->]
\tikzstyle{dashed edge}=[-, dashed, dash pattern=on 2pt off 0.5pt, draw=black]

\graphicspath{ {./img/} }

\makeatletter
\newcommand{\optionaldesc}[2]{%
  \phantomsection
  #1\protected@edef\@currentlabel{#1}\label{#2}%
}
\makeatother

\theoremstyle{definition}
\newtheorem{theorem}{Theorem}[section]
\newtheorem{corollary}[theorem]{Corollary}
\newtheorem{lemma}[theorem]{Lemma}
\newtheorem*{lemma*}{Lemma}
\newtheorem*{proposition*}{Proposition}
\newtheorem*{remark*}{Remark}
\newtheorem{proposition}[theorem]{Proposition}

\newtheorem{definition}[theorem]{Definition}
\newtheorem{example}[theorem]{Example}
\newtheorem{remark}[theorem]{Remark}

\newtheorem*{rep@theorem}{\rep@title}
\newcommand{\newreptheorem}[2]{%
	\newenvironment{rep#1}[1]{%
    \def\rep@title{#2 \ref{##1} (restated)}%
		\begin{rep@theorem}}%
		{\end{rep@theorem}}}
\newreptheorem{thm}{Theorem}

\DeclareMathOperator{\oddop}{Odd}
\newcommand{\odd}[1]{\oddop\mathopen{}\left(#1\right)\mathclose{}} 

\newcommand{\odds}[2]{\mathsf{Odd}_{#1}\left(#2\right)}

\newcommand{\symd}{\mathbin{\Delta}\xspace} 
\newcommand{\Symdi}[1]{\underset{\scriptstyle #1}{\scalebox{1.5}{$\symd$}}\,} 
\newcommand{\comp}[1]{\bar{#1}} 
\newcommand{\ld}{\lambda}
\newcommand{\sse}{\subseteq}
\newcommand{\abs}[1]{\left| #1 \right|} 
\newcommand{\FF}{\mathbb{F}}

\newcommand{\LOG}{labelled open graph}
\newcommand{\A}{{A}}

\newcommand{\I}{Id}
\newcommand{\Xlike}{\mathcal{X}}

\newcommand{\planar}{\mathcal{L}}

\newcommand{\Wpred}{\mathcal{W}_{\mathrm{pred}}}
\newcommand{\Vsucc}{\mathcal{V}_{\mathrm{succ}}}
\newcommand{\corrset}{L}
\newcommand{\outneighb}{T}

\newcommand{\ketbra}[2]{\ket{#1}\!\bra{#2}}

\renewcommand{\t}[1]{\ensuremath{^{\otimes #1}}}
\newcommand{\pow}[1]{\mathfrak{P}\!\left(#1\right)}

\title{Generating one-way computations with flow: flow-preserving rewriting that ignores the interpretation}

\author{Miriam Backens}
\affil{Université de Lorraine, CNRS, Inria, LORIA, F-54000 Nancy, France}
\date{}

\begin{document}

\maketitle

\begin{abstract}
	\noindent The one-way model is a universal model of quantum computation, driven by successive adaptive single-qubit measurements on an entangled resource state.
	Measurements are non-deterministic, yet if the computation satisfies one of several related families of conditions known as `flows', the computation can be made deterministic overall by modifying later measurements depending on the outcomes of earlier ones.
	Flow properties also enable efficient translation from one-way computations to circuits, motivating research into rewriting one-way computations while preserving the existence of flow.
	Existing approaches to flow-preserving rewriting are used for compilation or optimisation and preserve both the interpretation and the existence of flow.
	
	Here, we broaden our perspective to consider flow-preserving rewriting that does not necessarily preserve the interpretation, with applications to creating test instances for software that works with flow, as well as to generating ans\"atze for quantum machine learning.
	We show that a family of just three flow-preserving rewrite rules suffices to generate any diagram with flow from a trivial diagram with the desired number of inputs and outputs.
	This rule set is nearly the same as the complete set of flow- and interpretation-preserving rewrite rules for one-way computations in which all measurements are Pauli; and just a small subset of the flow- and interpretation-preserving rewrite rules for arbitrary measurements.
\end{abstract}

\section{Introduction}

The one-way model of measurement-based quantum computing (MBQC) is a universal model of quantum computation, analogous to quantum circuits.
It is a promising candidate for implementing photonic quantum computation as its structure of single-qubit measurements on an entangled resource state enables `front-loading' the entangling gates, which are particularly difficult to implement in photonics.
At the same time, the one-way model also has applications in theoretical quantum computing: its structure is much more flexible than that of quantum circuits, leading to advantages in circuit optimisation \cite{broadbent_parallelizing_2009,duncanGraph-theoretic2020}.
This additional flexibility has also been exploited for generating ans\"atze for quantum architecture search and quantum machine learning
\cite{ewen_application_2025,calderonMeasurementbasedQuantumMachine2025}.
Throughout these applications, flow properties \cite{danos_determinism_2006,browne_generalized_2007} are vital for ensuring deterministic implementability of a one-way computation, as well as efficient translation to quantum circuits.
A flow consists of a partial order on the measurements and a procedure for correcting undesired measurement outcomes, which must satisfy certain compatibility conditions.
There are different variants, from \emph{causal flow}, which characterises one-way computations that are very close to quantum circuits, via \emph{gflow} and \emph{extended gflow} to \emph{Pauli flow}.

The ZX-calculus is a graphical formalism for expressing and reasoning about quantum computations with applications to a variety of tasks at different levels of the quantum computational stack.
It also provides a common language for one-way computations and quantum circuits, and for translations between them \cite{duncanGraph-theoretic2020,backens_there_2021}.
While it is easy to translate a circuit into the ZX-calculus, translating a ZX-diagram into a quantum circuit is \#P-hard in general \cite{de_beaudrap_circuit_2022}.
Yet it is straightforward to bring a ZX-diagram into an `MBQC form' which means it can be interpreted as a one-way computation; and if this computation has flow, then there is also an efficient translation to a circuit \cite{duncanGraph-theoretic2020,backens_there_2021,simmons_relating_2021-1}.
Moreover, a recent work generalises flow properties to a broader class of ZX-diagrams than just those in MBQC-form under the name of `ZX-flow' \cite{kissingerZXFlowFlexibleCriterion2026}.
This means flow properties are important to anyone working with the ZX-calculus, whether or not they they are using the one-way model directly.

One major strength of the ZX-calculus are its complete equational theories: whenever two diagrams represent the same linear map, these equational theories allow one diagram to be transformed into the other entirely graphically \cite{backens_zx-calculus_2014,hadzihasanovic_two_2018,vilmart_near-minimal_2019}.
Yet the traditional equational theories for the ZX-calculus do not preserve flow, or even necessarily the form of a one-way computation.
Over the last years, researchers have thus developed flow-preserving rewrite rules: i.e.\ equations that preserve not only the interpretation of a ZX-diagram but also the MBQC form and the existence of flow \cite{duncanGraph-theoretic2020,backens_there_2021,simmons_relating_2021-1,mcelvanney_complete_2023,mcelvanney_flow-preserving_2023,backens_inserting_2025,kissingerZXFlowFlexibleCriterion2026}.
These have been used for a range of applications in compilation and optimisation \cite{duncanGraph-theoretic2020,holker_causal_2023,cao_multi-agent_2023}.

In this work, we broaden the concept of flow-preserving rewriting by dropping the requirement that the interpretation has to be preserved\footnote{This work is partly inspired by a talk of Rajarsi Pal at the Graphix Workshop 2024, which did not involve a completeness proof and -- as far as we can tell -- has never led to a publication.}.
Instead of transforming a given computation into a more preferred form, this changes the goal to that of generating one-way computations with flow.
These are needed as examples for research into the one-way model and as test cases for software that handles one-way computations, such as the library \emph{Graphix}\footnote{Available at \url{https://github.com/teamgraphix/graphix}.}.
Rewriting without preserving the interpretation also makes the increased flexibility of one-way computations (compared to quantum circuits) available to quantum architecture search and other quantum machine learning applications.
We therefore give a set of three rewrite rules, stated using the ZX-calculus, that preserve flow but not necessarily the interpretation.
Some of the rules come with side conditions that are required for flow-preservation.
This rule set is very close to the complete set of flow- and interpretation-preserving rewrite rules for one-way computations where all measurements are Pauli (so that the overall computation is within the stabiliser fragment) \cite{mcelvanney_complete_2023}.
We show that these three rules suffice to transform any one-way computation with Pauli flow or gflow to a trivial normal form with the same number of inputs and outputs, and conversely.
Removing any one of the rules would mean losing the ability to generate arbitrary diagrams with flow: in other words, the set is minimal.
We also discuss how to use the new rule set to generate computations with flow.

In the following, we first present background material about the ZX-calculus, the one-way model, and the Pauli flow and gflow properties in Section~\ref{s:background}.
We then state the set of flow- but not necessarily interpretation-preserving rewrite rules in Section~\ref{s:rule-set} and derive some further useful rewrite rules in Section~\ref{s:derived}.
The completeness and minimality proofs are in Section~\ref{s:completeness}.
Approaches to generating computations with flow are discussed in Section~\ref{s:generating}, followed by the conclusions in Section~\ref{s:conclusions}.

\section{Preliminaries}
\label{s:background}

We begin by introducing relevant background material.
First, we give the basics of the ZX-calculus, which will later be used to express the flow-preserving rewrite rules.
Then we introduce the one-way model and give a formal definition of flow.

\subsection{A very short introduction to the ZX-calculus}

Where quantum circuits consist of gates connected by straight horizontal wires, the ZX-calculus is more flexible.
Its three generators are green $Z$-spiders and red $X$-spiders, which can have arbitrary numbers of input and output wires, as well as square Hadamard nodes, which have one input and one output.
They represent the following linear maps:
\begin{align*}
	\input{zx-basics/green-spider.tikz} \quad&\rightsquigarrow\quad \ket{0}\t{n}\bra{0}\t{m} + e^{i\alpha} \ket{1}\t{n}\bra{1}\t{m} \\
	\input{zx-basics/red-spider.tikz} \quad&\rightsquigarrow\quad \ket{+}\t{n}\bra{+}\t{m} + e^{i\alpha} \ket{-}\t{n}\bra{-}\t{m} \\
	\begin{tikzpicture}
	\begin{pgfonlayer}{nodelayer}
		\node [style=hadamard] (0) at (0, 0) {};
		\node [style=none] (1) at (-1, 0) {};
		\node [style=none] (2) at (1, 0) {};
	\end{pgfonlayer}
	\begin{pgfonlayer}{edgelayer}
		\draw (1.center) to (2.center);
	\end{pgfonlayer}
\end{tikzpicture}
 \quad&\rightsquigarrow\quad \ketbra{+}{0} + \ketbra{-}{1}
\end{align*}
Spiders carry a \emph{phase label} $\alpha\in\mathbb{R}$, which can also be restricted without loss of generality to one of the intervals $[0,2\pi)$ or $(-\pi,\pi]$.
A spider with no label is interpreted as having phase $0$.

ZX-diagrams are composed in series by connecting the output wires of one diagram (reading from left to right, like circuits) to the input wires of another; this corresponds to the composition of the associated linear maps.
Diagrams can also be composed in parallel by juxtaposing them, this corresponds to the tensor product of the associated linear maps.
As long as the inputs and outputs of the overall diagram remain in the same order, a ZX-diagram can be seen as a graph: the relative positions of the generators in the plane of the paper do not matter, as long as the same components are connected in the same ways.
This is known as `only connectivity matters' and it allows vertical lines to meaningfully appear in diagrams, for example the controlled-$Z$ gate is represented in the ZX-calculus (up to normalisation) as:
\ctikzfig{zx-basics/cz-ex}
Non-zero scalar factors are commonly ignored in the ZX-calculus, particularly when working with unitary operators or unitary embeddings, where the normalisation is fixed by definition and any scalar factor of absolute value 1 has no physical effect anyway.

Since we will not need the full equational theory of the ZX-calculus here, we introduce only two straightforward equations that are used for implicit simplification of diagrams.
The \emph{spider rule} allows spiders of the same colour to be merged if they are connected by at least one wire, adding their phase labels.
The \emph{colour change rule} expresses that applying Hadamards to all inputs and outputs of a spider changes its colour:
\[
	\input{zx-basics/spider-fusion-Z.tikz} \qquad\qquad\qquad\qquad
	\input{zx-basics/colour-change.tikz}
\]
To avoid visual clutter in diagrams, a Hadamard that is connected to spiders on both ends is often denoted by a dashed blue line:
\ctikzfig{zx-basics/hadamard-edge}
For a full introduction to the ZX-calculus including the complete equational theories, see e.g.~\cite{vandewetering2020zxcalculus}.

\subsection{The one-way model}

A one-way computation proceeds via successive adaptive single-qubit measurements on an entangled resource state \cite{raussendorf_one-way_2001}.
This resource state is usually a graph state; a stabiliser state that corresponds to a simple graph as follows: for each vertex in the graph, a qubit is prepared in the state $\ket{+}$, and for each edge in the graph, a controlled-$Z$ gate is applied between the corresponding qubits.
The controlled-$Z$ gates are diagonal in the computational basis and thus commute with each other; hence their order is irrelevant.
In the ZX-calculus, a graph state is represented by a diagram of $Z$-spiders connected by Hadamard edges, with one dangling edge for each spider, for example the following represents a graph state on four qubits:
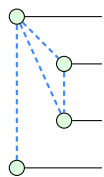
\ctikzfig{graph-state-ex}

Each measurement has a \emph{desired outcome} that drives the computation in the intended direction, and an undesired outcome.
Measurements are chosen in such a way that each undesired outcome differs from the desired outcome by a Pauli operator.
This makes it possible to use stabilisers of the underlying graph state to `correct' undesired measurement outcomes after they have happened, at the cost of (usually) introducing further Pauli operators on other qubits called \emph{by-products}.
If the by-products act non-trivially only on qubits that have not yet been measured, they can be absorbed into those measurements when they are later performed, and the one-way computation is overall deterministic (up to a Pauli product on the output qubits).
We now explain how to express a one-way computation and then define the different types of flow that formalise the correction procedure described above.

There are two ways of specifying a one-way computation.
The first way focuses on the time order and is called a \emph{measurement pattern}.
It consists of a sequence of commands for creating the entangled resource state and performing the single-qubit measurements as well as the corrections that are required for determinism \cite{danos_measurement_2007}.
The second way focuses on the entanglement structure; it consists of a \emph{labelled open graph} (defined below) with an additional phase angle for each measurement.
We will use the second type of specification, which allows visualisation and graphical reasoning using the ZX-calculus.

\begin{definition}
 A \emph{labelled open graph} is a tuple $(G,I,O,\ld)$ consisting of a simple graph $G=(V,E)$, subsets $I,O\sse V$ called the \emph{input} and \emph{output vertices}, and a function $\ld:\comp{O}\to\{X,Y,Z,XY,XZ,YZ\}$ called the \emph{measurement labelling}, where $\comp{O} := V\setminus O$ is the complement of $O$ in $V$.
\end{definition}

The `graph' part of a labelled open graph specifies the resource graph state.
The inputs are qubits that are not prepared in the state $\ket{+}$ but are instead provided by some previous computation and entangled with the other qubits.
Analogously, the outputs are those qubits that are not measured but can be passed on to later computational layers.
Finally, the measurement labelling gives information about how each non-output qubit is to be measured: $X,Y,Z$ denote the eigenbases of the three Pauli matrices, while the \emph{planar measurements} $XY,XZ,YZ$ are families of measurements in the planes of the Bloch sphere spanned by the eigenbases of two Paulis.
The labelled open graph encodes all the information that is relevant to verifying whether a computation can be implemented deterministically.
To compute the linear map implemented by a one-way computation, it is necessary to specify some additional information for each measurement, namely a phase angle $\alpha$ that indicates which measurement in a plane is the desired one (or which eigenstate is the desired outcome for the case of a Pauli measurement).

Measurements labelled $XZ, YZ$, or $Z$ will be called \emph{Z-like} and measurements labelled $XY, X$, or $Y$ will be called \emph{X-like}.
The term \emph{boundary vertices} refers to $I\cup O$ and vertices in $\planar:= \{v\in V\setminus (I\cup O)\mid \ld(v)\in\{XY,YZ,XZ\}\}$ are called \emph{internal} planar measurements.

In the ZX-calculus, a labelled open graph is represented by a graph state diagram with an additional input wire for each input vertex, and one of the following measurement effects for each measured qubit, depending on the measurement labels, where $\alpha\in\mathbb{R}$ and $a\in\{0,1\}$:
\begin{center}
	\begin{tabular}{c|c|c|c|c|c}
		$XY$ & $XZ$ & $YZ$ & $X$ & $Z$ & $Y$ \\ \hline
		\begin{tikzpicture}[baseline=-0.05]
	\begin{pgfonlayer}{nodelayer}
		\node [style=none] (0) at (-0.75, 0) {};
		\node [style=Z phase dot] (1) at (0.25, 0) {$\alpha$};
	\end{pgfonlayer}
	\begin{pgfonlayer}{edgelayer}
		\draw (0.center) to (1);
	\end{pgfonlayer}
\end{tikzpicture}
 & \begin{tikzpicture}[baseline=-0.05, scale=.8]
	\begin{pgfonlayer}{nodelayer}
		\node [style=none] (0) at (-0.75, 0) {};
		\node [style=Z phase dot] (1) at (0.25, 0) {$\frac{\pi}{2}$};
		\node [style=X phase dot] (2) at (1.5, 0) {$\alpha$};
	\end{pgfonlayer}
	\begin{pgfonlayer}{edgelayer}
		\draw (0.center) to (1);
		\draw (1) to (2);
	\end{pgfonlayer}
\end{tikzpicture}
 & \begin{tikzpicture}[baseline=-0.05]
	\begin{pgfonlayer}{nodelayer}
		\node [style=none] (0) at (-0.75, 0) {};
		\node [style=X phase dot] (1) at (0.25, 0) {$\alpha$};
	\end{pgfonlayer}
	\begin{pgfonlayer}{edgelayer}
		\draw (0.center) to (1);
	\end{pgfonlayer}
\end{tikzpicture}
 & \begin{tikzpicture}[baseline=-0.05]
	\begin{pgfonlayer}{nodelayer}
		\node [style=none] (0) at (-0.75, 0) {};
		\node [style=Z phase dot] (1) at (0.25, 0) {$a\pi$};
	\end{pgfonlayer}
	\begin{pgfonlayer}{edgelayer}
		\draw (0.center) to (1);
	\end{pgfonlayer}
\end{tikzpicture}
 & \begin{tikzpicture}[baseline=-0.05]
	\begin{pgfonlayer}{nodelayer}
		\node [style=none] (0) at (-0.75, 0) {};
		\node [style=X phase dot] (1) at (0.25, 0) {$a\pi$};
	\end{pgfonlayer}
	\begin{pgfonlayer}{edgelayer}
		\draw (0.center) to (1);
	\end{pgfonlayer}
\end{tikzpicture}
 & \begin{tikzpicture}[baseline=-0.05]
	\begin{pgfonlayer}{nodelayer}
		\node [style=none] (0) at (-1.25, 0) {};
		\node [style=Z phase dot] (1) at (0.25, 0) {$(-1)^a\frac{\pi}{2}$};
		\node [style=none] (2) at (2.75, 0) {};
		\node [style=X phase dot] (3) at (4.5, 0) {$(-1)^{a+1}\frac{\pi}{2}$};
		\node [style=none] (4) at (2, 0) {$=$};
	\end{pgfonlayer}
	\begin{pgfonlayer}{edgelayer}
		\draw (0.center) to (1);
		\draw (2.center) to (3);
	\end{pgfonlayer}
\end{tikzpicture}

	\end{tabular}
\end{center}

\subsection{Pauli flow, in its algebraic formulation}

The first type of flow to be defined was the type now known as causal flow \cite{danos_determinism_2006}, which applies to one-way computations that are very close to being circuits.
It was later generalised to gflow and Pauli flow~\cite{browne_generalized_2007}, which allow more complex types of corrections.
The latter two types differ from each other only in that Pauli flow allows the Pauli measurement labels $X,Y,Z$ while gflow requires all measurements to be planar.
The original definitions of gflow and Pauli flow rely on a partial order on the qubits, which restricts the time ordering of the measurement, and a \emph{correction function}, which for each measurement (implicitly) specifies the stabiliser that should be used to correct it.
The compatibility conditions that must be satisfied by these two mathematical objects are somewhat cumbersome to work with, particularly when it comes to verifying whether the insertion of a new qubit (with some given measurement label and neighbourhood) is flow-preserving.

The recent algebraic formulation of Pauli flow (which restricts straightforwardly to gflow) instead expresses the flow conditions in terms of two matrices called \emph{flow-demand} and \emph{order-demand matrix}, which are constructed from the adjacency matrix of the underlying graph \cite{mitosekAlgebraicFormulationPauli2026}.
The right inverse of the flow-demand matrix, called correction matrix, encodes the correction function.
The product of order-demand matrix and the correction matrix is required to be a directed acyclic graph (DAG), this DAG encodes the partial order.
Strictly speaking, the algebraic definition refers to a more restricted type of gflow or Pauli flow called `focused', yet this is without loss of generality as a labelled open graph has focused Pauli flow if and only if it has Pauli flow \cite{simmons_relating_2021-1} and it has focused gflow if and only if it has gflow \cite{mhalla_which_2014,backens_there_2021}.
We now give the relevant definitions, in which all matrices are defined over $\FF_2$.
Rows and columns of these matrices are assumed to be labelled by subsets of the vertices: given $S,T\sse V$, we will write `an $S\times T$ matrix' for a matrix whose rows are labelled by the elements of $S$ and whose columns are labelled by the elements of $T$.

\begin{definition}[{\cite[Definition~3.16]{mitosekAlgebraicFormulationPauli2026}}]
  Let $\Gamma=(G,I,O,\ld)$ be a labelled open graph.
    Then the \emph{extended adjacency matrix} $\A$ of $\Gamma$ satisfies $\A_{v,w} = 1$ if and only if either $\{v,w\}\in E$ or $v=w \wedge \ld(v)\in\{Y,XZ\}$.
\end{definition}

\begin{definition}[{\cite[Definition~3.4]{mitosekAlgebraicFormulationPauli2026}}]\label{def:flow-demand}
    Let $\Gamma = (G,I,O,\lambda)$ be a labelled open graph with extended adjacency matrix $\A$.
    The \emph{flow-demand matrix} $M_{\Gamma}$ is the $\comp{O}\times\comp{I}$ matrix with rows corresponding to non-outputs and columns corresponding to non-inputs, whose $v$-labelled row satisfies the following properties:
    \begin{itemize}
    	\item if $\lambda(v) \in \{ X, Y, XY \}$, then $M_{v,w} = \A_{v,w}$ for all $w\in\comp{I}$, and
    	\item if $\lambda(v) \in \{Z, YZ, XZ\}$, then $M_{v,v} = 1$ and $M_{v,w} = 0$ for all $w\in\comp{I}\setminus\{v\}$.
    \end{itemize}
\end{definition}

\begin{definition}[{\cite[Definition~3.5]{mitosekAlgebraicFormulationPauli2026}}]\label{def:order-demand}
     Let $\Gamma = (G,I,O,\lambda)$ be a labelled open graph with extended adjacency matrix $\A$.
     The \emph{order-demand matrix} $N_{\Gamma}$ is the $\comp{O}\times\comp{I}$ matrix whose $v$-labelled row satisfies the following properties:
     \begin{itemize}
     	\item if $\lambda(v) \in \{ X, Y, Z \}$, then $N_{v,w} = 0$ for all $w\in\comp{I}$;
     	\item if $\lambda(v) \in \{ XZ, YZ \}$, then $N_{v,w} = \A_{v,w}$ for all $w\in\comp{I}$; and
     	\item if $\lambda(v) = XY$, then $N_{v,v} = 1$ (provided that the $v$ column exists) and $N_{v,w} = 0$ for all $w\in\comp{I}\setminus\{v\}$.
     \end{itemize}
\end{definition}

Given a labelled open graph $\Gamma = (G,I,O,\lambda)$, any $\comp{I}\times\comp{O}$ matrix over $\FF_2$ is a potential \textit{correction matrix $C$} for $\Gamma$ \cite[Definition~3.6]{mitosekAlgebraicFormulationPauli2026}.
In the traditional definition of flow, such a matrix corresponds to a correction function $c:\comp{O}\to \pow{\comp{I}}$ (where $\pow{-}$ denotes the powerset) via the relationship $u \in c(v)$ if and only if $C_{u,v} = 1$.
In other words, the $v$-labelled column of $C$ encodes the set $c(v)$.
As it is often simpler to think about a correction function than about the column of a correction matrix, we will also use the correction function notion in the following sections.

Having defined the relevant matrices, we can now give the definition of Pauli flow.
If all measurements in the underlying labelled open graph are planar, this becomes a definition of gflow.

\begin{theorem}[Algebraic formulation of Pauli flow {\cite[Theorem~3.1]{mitosekAlgebraicFormulationPauli2026}}]\label{thm:algebraic-Pauli}
 Let $\Gamma = (G,I,O,\lambda)$ be a labelled open graph, $c$ a correction function on $\Gamma$, $M$ the flow-demand matrix of $\Gamma$, and $N$ the order-demand matrix of $\Gamma$.
 Then there exists a strict partial order $\prec$ such that $(c,\prec)$ is a focused Pauli flow on $\Gamma$ if and only if
 $M_{\Gamma}C = \I$, and
 $N_{\Gamma}C$ is the adjacency matrix of a directed acyclic graph,
 where $C$ is the correction matrix corresponding to $c$.
\end{theorem}

We will generally work with the partial order that is the transitive closure of the relationships encoded in the directed acyclic graph $N_{\Gamma}C$; this order is called the \emph{induced partial order}.
It has the property that all inputs are minimal and all outputs are maximal.

\begin{remark}\label{rem:outputs-partial-order}
 Note that the outputs appear in the partial order given by the directed acyclic graph $N_\Gamma C$, but only with relationships that arise from their appearances in correction sets (their appearances in odd neighbourhoods are irrelevant to the partial order).
 This will be important later.
\end{remark}

The following layering often associated with the partial order of a flow will be useful.

\begin{definition}[{\cite[Definition~4]{mhallaFinding2008}}]\label{def:layering}
 For a labelled open graph $\Gamma = (G,I,O,\ld)$ with $G=(V,E)$ and a partial order $\prec$ on $\Gamma$, define a \emph{layering} of the vertices via the following partition: let
 \[
  V_k^\prec = \begin{cases}
               \max_\prec(V) &\text{if } k = 0 \\
               \max_\prec(V\setminus(\bigcup_{i<k}V_i^\prec)) &\text{if } k > 0
              \end{cases}
 \]
 where $\max_\prec(A) = \{u\in A \mid \forall v\in A, \neg(u\prec v)\}$ is the set of the maximal elements of the set $A$ with respect to $\prec$.
 The \emph{depth} $d$ of the layering is the smallest integer $d$ such that $V_{d+1}^\prec = \emptyset$.
\end{definition}

If a partial order is associated with a flow, the depth of the corresponding layering is also referred to as the depth of the flow.
Whenever $\prec$ is the partial order of a Pauli flow, we have $O\sse V_0^\prec$ \cite{simmons_relating_2021-1}.
Moreover, if all measurements in $\Gamma$ are planar, then $V_0^\prec = O$ and $V_1^\prec$ is the set of vertices that can be corrected using only outputs \cite{backens_there_2021}.

\section{Flow-preserving rewrite rules}
\label{s:rule-set}

We now present the flow-preserving rewrite rules, summarised in Figure~\ref{fig:flow-preserving}, and discuss where these rules (or their interpretation-preserving equivalents) were first proved to be flow-preserving.
This rule set is quite similar to the complete flow- and interpretation-preserving rule set for one-way computations in the Clifford fragment \cite{mcelvanney_complete_2023}, the only changes are the presence of non-trivial phase variables in \eqref{ZL} and a (non--interpretation-preserving) simplification of \eqref{IO}.

As flow-preservation has to be proved for each direction of a rule separately, we use arrows instead of equation symbol to indicate this directionality.

\begin{figure}
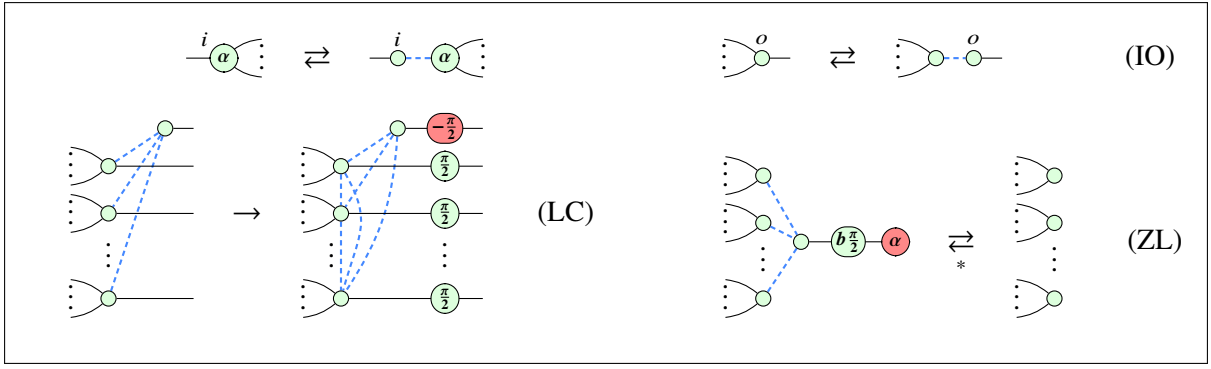

	\fbox{\begin{minipage}{.985\textwidth}
			\vspace{-.5em}
			\begin{subfigure}{.98\textwidth}
				\begin{align}\label{IO}\tag{IO}
					\input{rules/input-splitting-lhs.tikz} \quad\rightleftarrows\quad \input{rules/input-splitting-rhs.tikz} \qquad\qquad\qquad\qquad
					\input{rules/output-splitting-lhs.tikz} \quad\rightleftarrows\quad \input{rules/output-splitting-rhs.tikz}
				\end{align}
			\end{subfigure}
			
			\begin{subfigure}{0.49\textwidth}
			\begin{align}\label{LC}\tag{LC}
				\input{rules/LC-lhs.tikz} \quad\to\quad \input{rules/LC-rhs.tikz}
			\end{align}
			\end{subfigure}
			\begin{subfigure}{0.49\textwidth}
			\begin{align}\label{ZL}\tag{ZL}
				\input{rules/Z-like-rhs.tikz} \quad\underset{*}{\rightleftarrows}\quad \input{rules/Z-like-lhs.tikz}
			\end{align}
			\end{subfigure} \\
	\end{minipage}}
\caption{The complete flow-preserving (but not necessarily interpretation-preserving) rule set. The vertices marked $i$ in \eqref{IO} are inputs and the vertices marked $o$ are outputs. In both \eqref{IO} and \eqref{ZL}, $\alpha$ is an arbitrary phase angle; in the latter rule moreover $b\in\{0,1\}$. The top vertex in \eqref{LC} cannot be an input. The right-to-left application of \eqref{ZL} (marked with an asterisk) is flow-preserving only under certain conditions. All other rule directions always preserve the existence of gflow and Pauli flow.}
\label{fig:flow-preserving}
\end{figure}

\subsection{Basic flow-preserving rewrite rules}

Rule \eqref{IO} allows the extension or contraction of `dangling wires', which corresponds to inserting or removing degree-2 vertices along a path graph beginning at an input or output vertex.
That inserting or removing \emph{pairs} of vertices in this way preserves both gflow and interpretation is implicit in Refs.~\cite[Equation~(7)]{duncanGraph-theoretic2020} and~\cite[Section~4.1]{backens_there_2021}.
The gflow-preserving property of inserting or removing \emph{single} vertices in this way is also known \cite[Section~4.1]{backens_there_2021}; it is straightforward to see that this simpler variant does not generally preserve the interpretation.
Both sources referenced above work with gflow, hence without loss of generality, any new vertices arising from this rule may be considered to be $XY$-measured.

Rule \eqref{LC} is a local complementation, an equivalence operation that complements all edges among neighbours of some chosen vertex (i.e.\ deleting existing edges and inserting edges that were not previously there) and applies certain local Clifford operations to the chosen vertex and the neighbours \cite{van_den_nest_graphical_2004}.
Applying four successive local complementations about the same vertex has a trivial net effect, i.e.\ \eqref{LC} effectively provides its own inverse.
This is why, without loss of generality, the rule is expressed as having a unique direction.
Local complementation always preserves the existence of flow, assuming the central vertex is not an input -- see Ref.~\cite[Lemma~4.3]{backens_there_2021} for the gflow case and Ref.~\cite[Lemma~D.15]{simmons_relating_2021-1} for the Pauli flow case.
It preserves the interpretation if phase angles are updated suitably, which leads to problems if the transformation is applied to an output or to a neighbour of an output.
If the interpretation is not of interest, phase angles can be ignored instead.
In particular, if the interpretation is irrelevant, no particular care is needed for outputs or neighbours of an output.

Given an edge $uv$ in a graph, a sequence of three local complementations about first $u$, then $v$, then $u$ is known as a \emph{pivot}.
Local complementing about first $v$, then $u$, then $v$ has the same effect, so the pivot depends only on the edge, not its orientation.
As local complementations are flow-preserving, pivots preserve flow too \cite{backens_there_2021,simmons_relating_2021-1}; they are self-inverse.
The following is a simplified schematic representation of a pivot: it swaps the places of the two vertices $u$ and $v$, and moreover toggles edges connecting elements of three subsets of the vertices: joint neighbours of $u$ and $v$, neighbours of $u$ but not $v$, and neighbours of $v$ but not $u$.
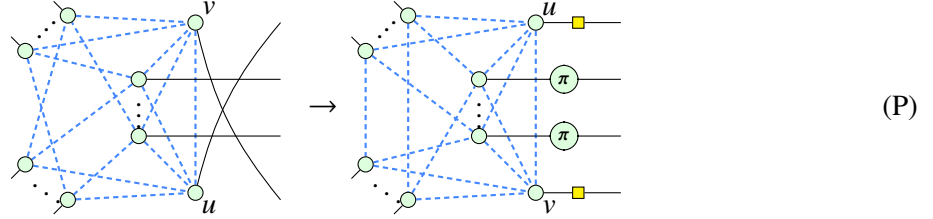
\begin{equation}\label{pivot}
	\input{pivot.tikz} \tag{P}
\end{equation}

Rule \eqref{ZL} is the deletion (if read left-to-right) or insertion (if read right-to-left) of a $Z$-like measured vertex.
Deletion of $Z$-like measurements always preserves the existence of flow; it preserves the interpretation only for certain angles \cite[Section~4.3]{backens_there_2021} \cite[Section~D.2]{simmons_relating_2021-1}.
Insertion of Pauli $Z$-measurements is always flow-preserving and preserves the interpretation if the angle is 0 \cite[Prop.~4.1]{mcelvanney_complete_2023}.
For insertion of $XZ$ and $YZ$-measurements, flow preservation depends on the choice of neighbourhood as shown in the below theorem.
This is stated for Pauli flow but also applies to gflow if all measurements are planar.

\begin{theorem}[{\cite[Theorems~4.3 and~5.2]{backens_inserting_2025}}]\label{thm:planar-Z-like-insertion}
 Let $\Gamma = (G,I,O,\ld)$ be a \LOG\ and let $\Gamma'$ be the \LOG\ that results from inserting a new $YZ$-or $XZ$-measured vertex $z$ with neighbourhood $S\sse V$.
 Then $\Gamma'$ has Pauli flow if and only if there exists a Pauli flow $(c,\prec)$ on $\Gamma$ and a set $K\sse\comp{I}$ such that:
 \begin{enumerate}
  \item\label{it:K-focused} $K$ is focused over $\comp{O}\setminus(S\cap\Xlike)$, and moreover $\comp{O}\setminus(S\cap\Xlike)$ is the largest set over which $K$ is focused.
  \item\label{it:K-S-even} $\abs{S\cap K} \equiv 0 \bmod 2$ if $\ld(z)=YZ$, and $\abs{S\cap K} \equiv 1 \bmod 2$ if $\ld(z)=XZ$.
  \item\label{it:order} If $\Wpred := \{w\in\comp{O}: \abs{c(w)\cap S}\equiv 1 \bmod 2\}$ and $\Vsucc := \planar\cap (K \cup (\odds{G}{K} \symd S))$, then all $w\in\Wpred$ and all $v\in\Vsucc$ satisfy $\neg(v\prec w \vee v = w)$;
  i.e.\ no vertex of $\Vsucc$ is equal to or precedes any vertex of $\Wpred$ for the original order $\prec$.
 \end{enumerate}
 If the properties hold, the focused Pauli flow on $\Gamma'$is given by the extension $c'$ of $c$ to domain $\comp{O}\cup\{z\}$ satisfying $c'(z) = K\cup\{z\}$,
 and the transitive closure of ${\prec} \cup \{(w,z)\mid w\in\Wpred\} \cup \{(z,v)\mid v\in\Vsucc \}$.
\end{theorem}

Insertion and deletion of $Z$-like measurements are inverse to each other: if inserting a $Z$-like measurement with a certain neighbourhood is flow-preserving, then so is deleting this measurement \cite[Corollary~4.5]{backens_inserting_2025}. Conversely, if a $Z$-like measurement with a certain neighbourhood has been deleted (which is always flow-preserving), then the same type of measurement can be reinserted with the same neighbourhood in a flow-preserving way.

A well-known combination of \eqref{LC} and \eqref{ZL} is the `pivot-and-delete' rule, which corresponds to applying a pivot and then deleting the resulting $YZ$-measurements:
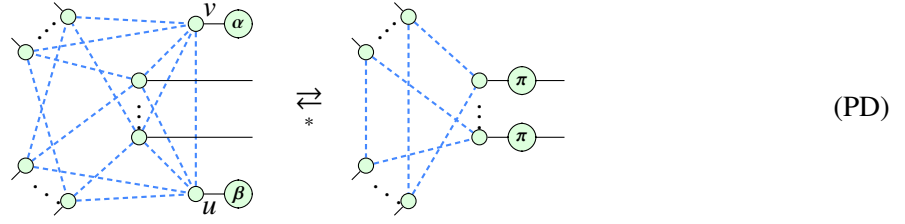
\begin{equation}\label{PD}
	\input{pivot-del.tikz} \tag{PD}
\end{equation}
This is interpretation-preserving if $\alpha,\beta\in\{0,\pi\}$ and certain adjustments are made to the phase angles of the neighbouring vertices \cite[Lemma~5.3]{duncanGraph-theoretic2020}.
It is always flow-preserving in the `deletion' direction; the flow-preservation conditions for the `insertion' direction follow from two applications of Theorem~\ref{thm:planar-Z-like-insertion}.

\subsection{Derived flow-preserving rules}
\label{s:derived}

Building on the above rules we now derive some further useful transformations and the conditions under which they are flow-preserving.
Like pivoting, these transformations will be convenient shorthands that preserve certain properties of the diagram that need not be preserved by the individual rules of Figure~\ref{fig:flow-preserving}.
The two rules of Lemmas~\ref{lem:toggle-edge} and~\ref{lem:merge-neighbours} in particular were previously used for generating computations with flow in Ref.~\cite{ewen_application_2025}, yet without a formal analysis of the conditions under which they preserve flow.

Throughout this section we consider labelled open graphs in which all measurements are planar, so the relevant flow property is gflow.
In a slight abuse of notation, we will extend the definition of the correction function to outputs via $c(v)=\emptyset$ for all $v\in O$.

First, we will need the following special case of Theorem~\ref{thm:planar-Z-like-insertion}, which is closely related (but not identical) to Theorem~6.3.1 of Ref.~\cite{mcelvanney_preservation_2025} and Corollary~5.6 of Ref.~\cite{backens_inserting_2025}.

\begin{corollary}\label{cor:simplified-planar-insertion}
 Let $\Gamma = (G, I, O, \lambda)$ be a labelled open graph which has gflow (i.e.\ all measurements are planar).
 Suppose $S = \{a,b\} \subseteq V$ are such that each of $a,b$ is either an output or $XY$-measured.
 Then a flow-preserving planar $Z$-like insertion with neighbourhood $S$ is possible if and only if there exists a focused Pauli flow $(c,\prec)$ for which no vertex is strictly between $a$ and $b$ in the induced partial order, i.e.\ $\not\exists u\in V : a\prec u\prec b \vee b\prec u\prec a$.

 Moreover, if $a\in c(b) \oplus b\in c(a)$, the newly-inserted vertex must be $XZ$-measured whereas if $a\notin c(b) \wedge b\notin c(a)$, the newly-inserted vertex must be $YZ$-measured\footnote{The combination $a\in c(b) \wedge b\in c(a)$ is not possible as this would imply a loop in the partial order.}.
\end{corollary}

\begin{proof}
  Take $K = c(a) \symd c(b)$ in Theorem~\ref{thm:planar-Z-like-insertion}, then condition~\ref{it:K-focused} (focusing) is satisfied.
Similarly, we satisfy either condition~\ref{it:K-S-even} of the $YZ$-insertion case or of the $XZ$-insertion case, as described in the corollary statement.
It remains to consider condition~\ref{it:order}.

Straightforwardly, the definition of $\Wpred$ reduces to $\Wpred := \{w\in\overline{O}: a \in c(w) \oplus b\in c(w) \}$.
For $\Vsucc$, we have
\begin{align*}
	\Vsucc &= \planar\cap \big((c(a) \symd c(b)) \cup (\odds{G}{c(a) \symd c(b)} \symd \{a,b\})\big) \\
	&= \planar\cap \big((c(a) \symd c(b)) \cup (\odds{G}{c(a)} \symd \{a\} \symd \odds{G}{c(b)} \symd \{b\})\big)
\end{align*}
thus $v\in\Vsucc$ implies $v\in c(a)\cup (\odds{G}{c(a)} \symd \{a\})$ or $v\in c(b)\cup (\odds{G}{c(b)} \symd \{b\})$.\footnote{The implication does not generally hold in the other direction as an element of $c(a)$ or $\odds{G}{c(a)}$ could be a boundary vertex, or some element could appear in both $c(a)$ and $c(b)$, in which case it need not be in $\mathcal{V}_{\mathrm{succ}}$. Yet the given implication is all we need.}
More concisely, as indicated by the `successor' part of the name $\Vsucc$, we have  $v\in\mathcal{V}_{\mathrm{succ}} \implies a\prec v \vee b\prec v$.
Similarly, $w\in\mathcal{W}_{\mathrm{pred}}$ means $w\prec a \vee w\prec b$.

Now suppose there exists a pair $v\in\mathcal{V}_{\mathrm{succ}}, w\in\mathcal{W}_{\mathrm{pred}}$ such that $v\prec w \vee v = w$.
There are four cases:
\begin{enumerate}
 \item $a\prec v, w\prec a$ which implies $w\prec a \prec v$ so $\neg (v\prec w \vee v = w)$ because $\prec$ is a partial order, contradicting the assumption $v\prec w \vee v = w$.
 \item $a\prec v, w\prec b$: then $v=w$ implies $a\prec v=w\prec b$ and $v\prec w$ implies $a\prec v\prec w\prec b$. In both cases, there exists some vertex strictly between $a$ and $b$ in the partial order.
 \item $b\prec v, w\prec a$: this is symmetric to case~2.
 \item $b\prec v, w\prec b$: this is symmetric to case~1.
\end{enumerate}
Hence if condition~\ref{it:order} is not satisfied, there is a vertex that sits strictly between $a$ and $b$ in the partial order.

Conversely, suppose there exists $u$ such that $a\prec u\prec b$.
Then there must be a vertex $u'$ with $u' = u \vee u'\prec u$ such that $u'\in c(a)\cup\odds{G}{c(a)}\setminus\{a\}$.
Similarly, there must be a vertex $u''$ with $u = u'' \vee u\prec u''$ such that $b\in c(u'')$.
Moreover, $u'\notin c(b)\cup\odds{G}{c(b)}\setminus\{b\}$ since $u'\prec b$; $u'\notin I$ since inputs do not appear in correction sets (by definitions) or in odd neighbourhoods (by focusing); and $u'\notin O$ since outputs are maximal in the induced partial order.
Therefore $u'\in\mathcal{V}_{\mathrm{succ}}$.
Similarly, $a\notin c(u'')$ since $a\prec u''$; therefore $u''\in\mathcal{W}_{\mathrm{pred}}$.
Thus, if there exists some vertex strictly between $a$ and $b$ in the partial order, then condition~3 of Theorem~\ref{thm:planar-Z-like-insertion} is false, as witnessed by the pair $u''\in\mathcal{W}_{\mathrm{pred}}$ and $u'\in\mathcal{V}_{\mathrm{succ}}$, which satisfy $u'=u'' \vee u'\prec u''$. 
The argument is analogous if instead $b\prec u\prec a$.

Combining the two directions, we find that insertion is possible if and only if there exists no vertex $u$ that comes strictly between $a$ and $b$ in the partial order.
\end{proof}

It can be useful to work specifically with labelled open graphs where all measurements are $XY$.
Moreover, one-way computations using only $XY$-measurements are universal \cite{danos_measurement_2007}.
The property of having only $XY$ measurements is not generally preserved by \eqref{LC} or \eqref{ZL}, but we show several ways of composing these rules so as to effect non-trivial transformations that preserve both flow and the property of having all measurement labels be $XY$.

First, the `edge toggling rewrite rule' either removes or inserts an edge between two $XY$-measured vertices.
This rule is flow-preserving if the edge is not required by the gflow (in the sense that one vertex is in the correction set of the other) and if the two vertices are `close together' in the partial order associated with the gflow.

\begin{lemma}\label{lem:toggle-edge}
 Let $\Gamma = (G, I, O, \lambda)$ be a labelled open graph which has gflow.
 Suppose $S = \{a,b\} \subseteq V$ are such that each of $a,b$ is either an output or $XY$-measured.
 If there exists a focused gflow $(c,\prec)$ for which $a\notin c(b) \wedge b\notin c(a)$ and no vertex is strictly between $a$ and $b$ in the partial order, then we can toggle the edge between $a$ and $b$ while preserving the existence of flow.

 In particular, an edge $\{a,b\}$ can always be toggled while preserving the existence of flow if one of the following conditions holds:
 \begin{itemize}
  \item both endpoints are inputs, i.e.\ $\{a,b\}\subseteq I$, or
  \item both endpoints are outputs, i.e.\ $\{a,b\}\subseteq O$, or
  \item at least one endpoint is a simultaneous input and output, i.e.\ $a\in I\cap O \vee b\in I\cap O$.
 \end{itemize}
\end{lemma}

\begin{proof}
 By Corollary~\ref{cor:simplified-planar-insertion}, the $YZ$-insertion in the first step below is flow-preserving under the conditions of the lemma:
 \[
  \input{insertion0.tikz}
  \quad\overset{\eqref{ZL}}{\to}\quad \input{insertion-YZ1.tikz}
  \quad\overset{\eqref{LC}}{\to}\quad \input{insertion-YZ2.tikz}
  \quad\overset{\eqref{ZL}}{\to}\quad \input{insertion-YZ3.tikz}
 \]
 The second step uses a local complementation and the third step uses $YZ$-deletion; both always preserve the existence of flow.
 The argument is analogous if initially there is no edge $\{a,b\}$.

 Note that while we have drawn measurement effects for both $a$ and $b$ and we have not specifically drawn input wires, both vertices may be inputs and/or outputs.
 Inputs are mutually incomparable and do not appear in correction sets, so the conditions of the lemma are trivially satisfied.
 Similarly, outputs are mutually incomparable and have empty correction sets, so again the conditions of the lemma are trivially satisfied.

 Finally, if one of the vertices (without loss of generality $a$) is a simultaneous input and output, i.e.\ $a\in I\cap O$, then $a$ does not appear in correction sets and it has an empty correction set.
 Therefore, by Remark~\ref{rem:outputs-partial-order}, $a$ is incomparable to all vertices and the conditions of the lemma are satisfied.
\end{proof}

Secondly, it is possible to merge two neighbouring qubits in a gflow-preserving way if one of them is in the correction set of the other, at least one of the two vertices is internal, and they are not separated in the partial order of the gflow.

\begin{lemma}\label{lem:merge-neighbours}
 Let $\Gamma = (G, I, O, \lambda)$ be a labelled open graph which has gflow.
 Suppose $S = \{a,b\} \subseteq V$ where $a\sim b$, each of $a,b$ is either an output or $XY$-measured, and one of the two vertices is internal, i.e.\ $\{a,b\}\setminus (I\cup O) \neq \emptyset$.
 If there exists a focused gflow $(c,\prec)$ for which $a\in c(b) \oplus b\in c(a)$ and no vertex is strictly between $a$ and $b$ in the partial order, we can merge vertices $a$ and $b$ while preserving the existence of gflow.
\end{lemma}
\begin{proof}
 Without loss of generality assume $b\notin I\cup O$, i.e.\ $b$ is internal, otherwise swap the designations of $a$ and $b$.
 By Corollary~\ref{cor:simplified-planar-insertion}, the $XZ$-insertion in the first step below is flow-preserving under the conditions of the lemma:
 \begin{equation}\label{merge}
  \input{insertion0.tikz}
  \quad\overset{\eqref{ZL}}{\to}\quad \input{insertion-XZ1.tikz}
  \quad\overset{\eqref{LC}}{\to}\quad \input{insertion-XZ2.tikz}
  \quad\overset{\eqref{pivot}}{\to}\quad \input{insertion-XZ3.tikz}
  \quad\overset{\eqref{ZL}}{\to}\quad \input{insertion-XZ4.tikz}
 \end{equation}
 The second step uses a local complementation, the third step uses pivoting, and the last step uses $YZ$-deletions, each of which always preserve the existence of flow.
 The neighbours of the remaining vertex at the end are exactly $(N_G(a)\symd N_G(b))\setminus\{a,b\}$.
 Note that while we have drawn measurement effects for both $a$ and $b$ and we have not drawn input wires, $a$ may be an input or output in the above argument.
\end{proof}

\begin{remark}\label{rem:vertex-splitting}
	The flow-preservation conditions for the reverse of the merge rule are somewhat more complicated, as the final step of \eqref{merge} becomes two $Z$-like insertions instead of deletions.
	Indeed, with slightly different phase labels (which, in the gflow setting, are irrelevant to the question of whether there is flow), the reverse of the final two steps of \eqref{merge} is known as `vertex splitting' in the literature:
	\ctikzfig{vertex-splitting}
	The exact conditions for flow-preservation of vertex splitting can be found in Theorem~5.5 of Ref.~\cite{backens_inserting_2025}; they follow from the $Z$-like insertion conditions.
	The remaining steps of the right-to-left direction of \eqref{merge} are a local complementation and a $Z$-like deletion, which are unconditionally flow-preserving.
\end{remark}

We will consider a useful special case of the reverse of \eqref{merge}, applying to outputs, in Proposition~\ref{prop:output-splitting}.

In the above lemma, if instead $\{a,b\}\sse I\cup O$, then since one must be in the other's correction set, we can assume without loss of generality that $a\in I, b\in O$.
In this situation, we cannot apply the pivot operation that is part of the derivation \eqref{merge}.
Nevertheless, if one of $a,b$ has a neighbourhood of size~1, a similar operation can be performed using \eqref{IO}.

The following result was proved using a different rule set that preserves both flow and interpretation, but it uses only rules \eqref{ZL} and \eqref{LC}, so it can be applied in our setting as well.
We only need the special case where $a$ and $b$ are both outputs.

\begin{lemma}[\cite{backens_completeness_2026}]\label{lem:insert-cnot}
	Let $\Gamma = (G,I,O,\ld)$ be a labelled open graph which has gflow.
	Suppose $a,b\in O$ are not adjacent to each other, then the following operation is flow-preserving, where $N_a := N_G(a)$ is the neighbourhood of $a$, $N_b := N_G(b)$ is the neighbourhood of $b$, and the two new vertices are both $XY$-measured:
	\[
		\input{symdi0.tikz} \quad\to\quad \input{symdi3b.tikz}
	\]
\end{lemma}

\begin{lemma}\label{lem:symmetric-difference}
	Let $\Gamma = (G,I,O,\ld)$ be a labelled open graph which has gflow.
	Suppose $a,b\in O$ are not adjacent to each other.
	Then replacing the neighbourhood of $a$ by the symmetric difference $N_G(a)\symd N_G(b)$ preserves the existence of gflow.
\end{lemma}
\begin{proof}
	This follows by simplifying the end result of Lemma~\ref{lem:insert-cnot} in non-interpretation-preserving ways.
	To save space, write $N' := N_a\symd N_b$.
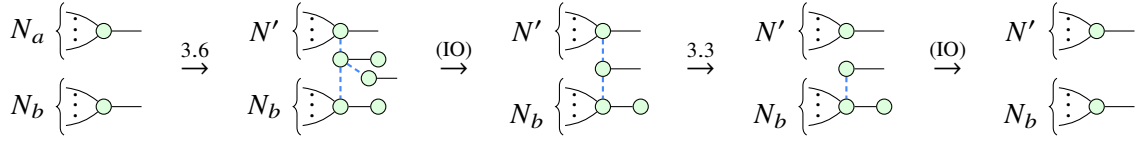
	\[
		\input{symdi0.tikz}
		\quad\overset{\ref{lem:insert-cnot}}{\to}\quad \input{symdi3c.tikz}
		\quad\overset{\eqref{IO}}{\to}\quad \input{symdi4.tikz}
		\quad\overset{\ref{lem:toggle-edge}}{\to}\quad \input{symdi5.tikz}
		\quad\overset{\eqref{IO}}{\to}\quad \input{symdi6.tikz}
	\]
	The application of Lemma~\ref{lem:toggle-edge} is flow-preserving as the endpoints of the edge are both outputs.
\end{proof}

\section{Completeness}
\label{s:completeness}

We now show that any MBQC pattern with Pauli flow (or any ZX-diagram with Pauli flow) can be brought into a trivial normal form in a flow-preserving way.
This then implies that every diagram can be generated from this trivial normal form.

\begin{definition}
 We say a labelled open graph $(G,I,O,\ld)$ is \emph{trivial} if $V=O$ and $E=\emptyset$.
 Similarly, a ZX-diagram is \emph{trivial} if it consists of a collection of wires that `inject' the set of inputs into the set of outputs, and if furthermore every output not connected to an input is instead connected to a unique degree-1 green spider.
\end{definition}

A trivial labelled open graph trivially has Pauli flow since every vertex is an output.
The trivial ZX-diagram is unique only up to a permutation of the outputs.
We now show that this is not a problem as it is possible permute the outputs arbitrarily using the rules of Figure~\ref{fig:flow-preserving}.

\begin{lemma}\label{lem:permute-outputs}
	Let $\Gamma = (G, I, O, \lambda)$ be a labelled open graph where every vertex is an output and there are no edges, i.e.\ with $G=(V,E)$ we have $V=O$ and $E=\emptyset$.
	Suppose that the inputs and outputs are separately ordered (as is the case in a ZX-diagram).
	Then it is possible permute the outputs arbitrarily while preserving the existence of flow.
\end{lemma}
\begin{proof}
	We show how to transpose two outputs.
	As any permutation can be built from transpositions, this implies the desired result.
	The proof is shown for two vertices in $I\cap O$, but it works the same way if one or both vertices are not inputs.
	\begin{align*}
		\input{permutation0.tikz}
		\quad\overset{\eqref{IO}}&{\to}\quad \input{permutation1.tikz}
		\quad\overset{\ref{lem:toggle-edge}}{\to}\quad \input{permutation2.tikz}
		\quad\overset{\eqref{pivot}}{\to}\quad \input{permutation3.tikz}
		\quad\overset{\eqref{ZL}}{\to}\quad \input{permutation4.tikz} \\
		\quad\overset{\ref{lem:toggle-edge}}&{\to}\quad \input{permutation5.tikz}
		\quad\overset{\eqref{IO}}{\to}\quad \input{permutation-final.tikz}
	\end{align*}
	It is straightforward to see that the conditions for Lemma~\ref{lem:toggle-edge} hold in both situations where it is used.
\end{proof}

\subsection{Trivialisation procedure}
\label{s:procedure}

Suppose $\Gamma = (G,I,O,\ld)$ is an arbitrary labelled open graph with Pauli flow.
Recall that if a vertex $u$ is a combined input and output, i.e.\ $u\in I\cap O$, then it can be disconnected from the rest of the graph while preserving flow using Lemma~\ref{lem:toggle-edge}.
In the following, we therefore assume that any vertices that do have neighbours are not in $I\cap O$.

Then $\Gamma$ can be made trivial in a flow-preserving way using the following five steps.

\begin{enumerate}
 \item \emph{Remove all Pauli measurements.}
  First, relabel all inputs to $XY$, this is flow-preserving as inputs do not appear in any correction sets (not even their own).
  Then:
  \begin{itemize}
   \item Delete all $Z$-measurements using $Z$-like deletion.
   \item Local complement and delete all $Y$-measurements\footnote{This is basically Ref.~\cite[Lemma~5.2]{duncanGraph-theoretic2020}.}.
   \item Note that every $X$-measurement has at least one neighbour (because it must be in the odd neighbourhood of its own correction set).
   If one of the neighbours is an internal vertex, pivot on this pair and then delete the former $X$-measurement (which has turned into $Z$).
   Otherwise, pick one of the boundary neighbours and extend its dangling wire so it is no longer a boundary; then proceed as before\footnote{This is fairly similar to Ref.~\cite[Lemma~5.3 and Equation~(7)]{duncanGraph-theoretic2020}, though we allow arbitrary measurement labels whereas in the reference, implicitly, everything is $XY$-labelled.}.
  \end{itemize}
  In this way, all Pauli measurements can be removed while preserving the existence of flow.
  The labelled open graph now has only planar measurements, i.e.\ it has gflow.

 \item \emph{Remove all $Z$-like planar measurements}, i.e.\ all $YZ$ and $XZ$ measurements, using flow-preserving $Z$-like deletion.
  Afterwards, the labelled open graph contains only $XY$-measurements.
  
 \item \emph{Merge internal vertices into outputs.}
 Suppose there exists at least one measured non-input vertex, otherwise move directly to the next step.
 Find a focused gflow $(g,\prec)$.
 
 Choose a measured non-input vertex $x$ that is maximal in $\prec$ (among the non-outputs), then $g(x)\sse O$.
 We can find $o\in g(x)$ such that $o\sim x$ since $x$ is $XY$-measured.
 Now, $S = \{o,x\}$ straightforwardly satisfies the conditions of Lemma~\ref{lem:merge-neighbours}, so we can merge the two vertices.
 Update the gflow accordingly.
 This step can be repeated as long as there there remain any measured non-input vertices.
 
 \item \emph{Merge inputs with outputs.}
 Now $I\cup O = V$ and by Lemma~\ref{lem:toggle-edge} we may assume the graph is bipartite with partition $I\setminus O$ and $\overline{I\setminus O}$; this property will be preserved throughout this step.
 (We continue to ignore elements of $I\cap O$ without loss of generality.)
 
 If $I\sse O$ (so the first part is empty) and $E=\emptyset$, the labelled open graph is trivial and we are done.
 Otherwise, find a focused gflow $(g,\prec)$.
 Note that this gflow must have depth 1: inputs do not appear in correction sets and every vertex is an input or an output, so every non-output is corrected using only output vertices.
 Moreover, focusing implies that the $XY$-measured input vertices do not appear in odd neighbourhoods of other vertices' correction sets.
 This means there can be no vertex that appears between two others in the partial order.
 Thus:
 \begin{enumerate}
 	\item\label{step:toggle} If there exists a pair of neighbours $i\in I, o\in O$ such that $o\notin g(i)$, then the edge $i\sim o$ can be removed by Lemma~\ref{lem:toggle-edge}.
 	
 	\item\label{step:IO} If there exists an input or an output with a single neighbour, this vertex and its neighbour can be merged using \eqref{IO}.
 	Afterwards, any edges connecting the merged vertex (which is now a simultaneous input and output) to any other vertex can be removed using Lemma~\ref{lem:toggle-edge}.
 	The vertex can then be removed from consideration.
 	
 	\item\label{step:symdi} If neither of the above hold, pick $i\in I$ and $o\in g(i)\cap N_G(i)$.
 	Then, for every $o'\in N_G(i)\setminus\{o\}$, apply Lemma~\ref{lem:symmetric-difference} to the pair $o',o$.
 	At the end of this process, $o$ will be the only neighbour of $i$, so we can go to Step~\ref{step:IO}.
 \end{enumerate}
  Each Step~\ref{step:toggle} reduces the number of edges and each Step~\ref{step:IO} reduces the number of vertices in $I\setminus O$.
  While Step~\ref{step:symdi} may increase the number of edges, it always leads immediately to Step~\ref{step:IO}.
  Thus the tuple $(\abs{I\setminus O}, \abs{E})$ strictly decreases in lexicographical ordering, and the procedure terminates with $I\sse O$ and no edges remaining in the graph.
 
 \item \emph{Permute outputs if needed.} This uses Lemma~\ref{lem:permute-outputs}.
\end{enumerate}

It is possible to push back the point at which a gflow is needed in the trivialisation algorithm by pivot-and-deleting pairs of adjacent internal $XY$-measurements while such pairs exist.
This means only having to run the flow-finding algorithm when the total number of qubits remaining is at most $\abs{I} + 2\abs{O}$.

By reversing the steps of the above algorithm, it is possible to generate any labelled open graph with Pauli flow.
Moreover, as Pauli measurements are removed in the first step of the trivialisation algorithm, the same approach can be used to generate arbitrary diagrams with gflow by simply not inserting any Pauli-measured vertices.
Similarly, if a diagram with only $XY$-measurements is desired, this also naturally follows from inverting the steps of the trivialisation algorithm.
We will consider a generating algorithm for diagrams with flow in Section~\ref{s:generating}.

\subsection{Minimality}

Removing any one of the rules of Figure~\ref{fig:flow-preserving} would mean losing the ability to make trivial (or, conversely generate) arbitrary diagrams with flow: in other words, the set is minimal.
This is straightforward to see: \eqref{IO} is the only rule that can turn an element of $I\cap O$ into two vertices of which one is an input and non-output and the other is an output and non-input.
The $Z$-like insertion and deletion rule \eqref{ZL} is the only rule that can change the number of connected components of the diagram.

Finally, the local complementation rule \eqref{LC} is the only rule that can change measurement labels on vertices, which is the only way of producing $X$- or $Y$-measurements.
It is also required in the planar-only setting as the other rules cannot change the topology of the $XY$-skeleton of a diagram: by this `skeleton' we mean the diagram that results from considering only the $XY$-measurements and outputs; this induced subdiagram always has gflow \cite{backens_there_2021}.
Rule \eqref{ZL} clearly does not affect the skeleton; and rule \eqref{IO} can only produce `dangling paths' leading to inputs or outputs, which do not change the topology of the skeleton.

\section{Generating one-way computations with flow}
\label{s:generating}

We will focus here on generating one-way computations with flow where all measurements are of type~$XY$.
These measurements suffice to implement any unitary embedding \cite{danos_measurement_2007} and the restriction will simplify the initial analysis.
Moreover, since the trivialisation procedure of Section~\ref{s:procedure} begins by removing all measurements with labels other than $XY$, we can conversely generate any one-way computation by first producing an $XY$-only skeleton, and then adding other measurements afterwards.

\subsection{A simplified approach for generating $XY$-only computations}

Starting from the trivial computation with $n_I$ inputs and $n_O$ outputs, we will generate a one-way computation with $n$ qubits in total.
To begin, we will show how to generate such a diagram without needing to keep track of a gflow via a special variant of vertex splitting (cf. Remark~\ref{rem:vertex-splitting}) that applies only in the vicinity of the outputs, and is always flow-preserving.

A vertex splitting operation involves two $YZ$-insertions, as can be seen by reversing the merging procedure of Lemma~\ref{lem:merge-neighbours}.
By choosing a suitable ordering of the insertions, one insertion will be with a single neighbour that is either $XY$-measured or an output; this insertion is always flow-preserving \cite[Corollaries~4.6 and~B.2]{backens_inserting_2025}.
The idea of applying the splitting rule in the vicinity of outputs is to choose the neighbourhood $W$ for the second $YZ$-insertion as follows:
\begin{itemize}
	\item $W\cap\comp{O} = \odds{G}{K}$ for some $K\sse O\setminus I$, i.e.\ the non-output neighbours of the new vertex $z$ must be exactly the odd neighbourhood of some subset of outputs.
	The set $K$ cannot contain any inputs as its elements will be in the correction set of the new vertex.
	Indeed, setting the correction set of the new vertex to be $\{z\}\cup K$ ensures that it produces trivial by-products on all non-outputs.
	
	\item The new vertex may also have output neighbours.
	Yet since $z$ needs to be $YZ$-measured, we require $\abs{W\cap K}\equiv 0 \bmod 2$ so that $z\notin\odd{\{z\}\cup K}$.
\end{itemize}
Given these two conditions, the new vertex will be inserted at the top of the partial order, i.e.\ maximal among measured vertices, so gflow is preserved without further conditions.
In particular, it is not necessary to know a gflow on the original labelled open graph to decide whether the insertion is flow-preserving.
This is formalised in the following proposition, where we separate out one element of the set of outputs as $K = L\cup\{o\}$, where $o$ represents the output that is being split.

\begin{proposition}\label{prop:output-splitting}
	Let $\Gamma = (G,I,O,\ld)$ be a labelled open graph where $\ld(\comp{O}) = \{XY\}$ which has gflow.
	Suppose $o\in O\setminus I$ has neighbourhood $N := N_G(o)$.
	Pick $\corrset\sse O\setminus(I\cup\{o\})$ and $\outneighb\sse O\setminus\{o\}$ such that $o\notin\odds{G}{\corrset}$ and $\abs{\corrset\cap (\odds{G}{\corrset}\symd \outneighb\symd N)}\equiv 0 \bmod 2$.
	Define $W := \odds{G}{\corrset}\symd \outneighb\symd N$.
	Then splitting $o$ over $W$ as in the diagram below is gflow-preserving:
	\[
	\input{vertex-splitting-output0.tikz} \quad\to\quad \input{vertex-splitting-output5.tikz}
	\]
\end{proposition}

\begin{remark}
	The requirement $o\notin\odds{G}{\corrset}$ in the above proposition ensures that $o\notin W$, as it does not make sense to split a vertex over a set containing itself.
	The property $\abs{\corrset\cap W}\equiv 0\bmod 2$ ensures that the second $Z$-like insertion in the proof below is a $YZ$-insertion (not an $XZ$-insertion).
\end{remark}

\begin{proof}
	The transformation is basically a reverse of that in the proof of Lemma~\ref{lem:merge-neighbours}.
	We have:
	\begin{multline*}
		\input{vertex-splitting-output0.tikz}
		\quad\overset{\eqref{ZL}}{\to}\quad \input{vertex-splitting-output1.tikz}
		\quad\overset{\eqref{ZL}}{\to}\quad \input{vertex-splitting-output2.tikz}
		\quad\overset{\eqref{pivot}}{\to}\quad \input{vertex-splitting-output3.tikz} \\
		\quad\overset{\eqref{LC}}{\to}\quad \input{vertex-splitting-output4.tikz}
		\quad\overset{\eqref{ZL}}{\to}\quad \input{vertex-splitting-output5.tikz}
	\end{multline*}
	It is straightforward to show that the first $YZ$-insertion is unconditionally gflow-preserving.
	As $XZ$-deletion, pivoting, and local complementation always preserve gflow, it only remains to consider the second $YZ$-insertion.
	
	For ease of reference, denote by $y$ the first newly-inserted $YZ$-measured vertex (whose only neighbour initially is $o$) and denote by $z$ the second new $YZ$-measured vertex, whose neighbours are $W\cup\{y\}$.
	Write $G'=(V',E')$ for the graph after the two $YZ$-insertions so that $V' := V\cup\{y,z\}$ and
	\[
	E' = E \cup \{\{o,y\}, \{y,z\}\} \cup \{\{z,w\}\mid w\in W\}.
	\]
	Note that $o\notin\odds{G}{\corrset}$ and $o\notin \outneighb$ by assumption, and $o\notin N$ as we are working with simple graphs, therefore $o\notin W = \odds{G}{\corrset}\symd \outneighb\symd N$.
	
	For any $v\in \corrset$, the updated neighbourhood $N_{G'}(v)$ equals either $N_G(v)$ or $N_G(v)\cup\{z\}$.
	This is because all new edges involve at least one of $y$ and $z$, and $N_{G'}(y)=\{o,z\}$.
	Moreover, $z\in N_{G'}(v)$ if and only if $v\in W$.
	Since $\abs{\corrset\cap W}\equiv 0 \bmod 2$ by assumption, this means
	\[
	\odds{G'}{\corrset} = \odds{G}{\corrset} \symd \Symdi{v\in \corrset\cap W} \{z\} = \odds{G}{\corrset}.
	\]
	Then:
	\begin{align*}
		\odds{G'}{\corrset\cup\{o, z\}}
		&= \odds{G'}{\corrset} \symd N_{G'}(o) \symd N_{G'}(z) \\
		&= \odds{G}{\corrset} \symd (N\cup\{y\}) \symd (W\cup\{y\}) \\
		&= W\symd N\symd \outneighb \symd N \symd \{y\} \symd W \symd\{y\} \\
		&= \outneighb,
	\end{align*}
	where the second step uses the property that the odd neighbourhood of $\corrset$ does not change, and the third step uses the involutive property of symmetric difference as well as the fact that $y\notin V$, which implies $y\notin N$ and $y\notin W$.
	
	Suppose the gflow on $\Gamma$ is $(g,\prec)$.
	Define $\Gamma' = (G',I,O,\ld')$, where $\ld'$ is the extension of $\ld$ to domain $V'\setminus O$ that satisfies $\ld'(y)=\ld'(z)=YZ$.
	Let
	\[
	g'(v) := \begin{cases}
		\{y\} &\text{if } v = y \\
		\corrset\cup\{o, z\} &\text{if } v = z \\
		g(v) &\text{otherwise}
	\end{cases}
	\]
	Note that $g'(z)\setminus\{z\}\sse O$ and $g'(y)\setminus\{y\}=\emptyset\sse O$.
	Moreover, $\odds{G'}{g'(y)} = \{o,z\}$ and $\odds{G'}{g'(z)} = \outneighb \sse O$.
	This means $y$ precedes $z$, but otherwise the only successors of $y$ or $z$ are outputs.
	Hence it is always possible to extend $\prec$ to a partial order $\prec'$ such that $(g',\prec')$ is a gflow, because no cycle can be created when inserting $y$ and $z$ at the very top of the partial order on $\comp{O}$.
\end{proof}

\begin{corollary}
	If $\Gamma$ in the above proposition has the property that there are no edges between outputs, then the conditions on $\corrset$ and $\outneighb$ simplify to $\abs{\corrset\cap \outneighb}\equiv 0\bmod 2$; so in particular the choice `$\outneighb=\emptyset$ and $\corrset$ arbitrary' always works.
	
	If there are initially no edges between outputs and if $\outneighb=\emptyset$, then $W\sse\comp{O}$ and hence Proposition~\ref{prop:output-splitting} does not introduce any edges between outputs.
	Therefore it is possible to work entirely in the setting where there are no edges between outputs.
\end{corollary}

\begin{remark}
	If $\corrset = \outneighb = \emptyset$, then Proposition~\ref{prop:output-splitting} corresponds to \eqref{IO}, except the proposition cannot be applied to an element of $I\cap O$.
\end{remark}

Now, reversing the trivialisation procedure of Section~\ref{s:procedure} suggests the following two-step generation procedure:
\begin{enumerate}
	\item Permute outputs as needed using Lemma~\ref{lem:permute-outputs}.
	Unfuse a subset of the vertices in $I\cap O$ into separate inputs and outputs using \eqref{IO}, then modify the connectivity of the resulting diagram using Lemmas~\ref{lem:toggle-edge} and~\ref{lem:symmetric-difference}.
	(Any remaining elements of $I\cap O$ will participate in the computation in somewhat trivial ways.)
	
	\item Apply Proposition~\ref{prop:output-splitting} until the desired total number of qubits is achieved.
\end{enumerate}

\begin{remark}
	Interestingly, here, it is the `merge' direction used in the trivialisation procedure that requires knowing the flow.
	The flow ensures that the merge operation of Lemma~\ref{lem:merge-neighbours} is applied to a measured vertex that is maximal in the partial order, and with an output neighbour that is in the correction set of the chosen measured vertex.
	On the other hand, the `split' direction does not require keeping track of the flow as it is possible to add a new vertex to the `top of the partial order' just below the outputs without knowing what the partial order is.
\end{remark}

\subsection{Analysis of the simplified approach}

Another advantage of creating one-way computations via these insertions at the end is that the number of options for the next insertion does not increase with the size of the diagram that has already been generated.
For a general $Z$-like insertion into a diagram with $n$ qubits, there are $2^n$ options for the neighbourhood (the set denoted $S$ in Theorem~\ref{thm:planar-Z-like-insertion}).
Many of those neighbourhoods will not be compatible with a flow-preserving insertion, but we are not aware of any general method of determining which neighbourhoods work without running through the conditions of Theorem~\ref{thm:planar-Z-like-insertion} for every such $S$.
The one exception is Corollary~\ref{cor:simplified-planar-insertion}, which considers the case where the neighbourhood consists of exactly two vertices, each of which is either $XY$-measured or an output.
There are some other simple cases where $YZ$-insertion is always flow-preserving, such as $S = \emptyset$ or $S$ consisting of a single $XY$-measured vertex \cite[Corollary~4.6]{backens_inserting_2025}, as well as $S\sse I$ or $S\sse O$ \cite[Corollaries~B.2 and~B.3]{backens_inserting_2025}, yet these cover only a small fraction of the possible neighbourhoods.

Contrasting with this, let $m:=\abs{O}$, then in Proposition~\ref{prop:output-splitting}:
\begin{itemize}
	\item There are $m$ choices for the vertex $o$.
	
	\item If $o$ has no output neighbours (i.e.\ $N_G(o)\cap O = \emptyset$), then there are $2^{m-1}$ choices for $\corrset$.
		If $o$ does have output neighbours, there are $2^{m-2}$ choices for $\corrset$.
		
		To see this, note that $\abs{\pow{O\setminus\{o\}}} = 2^{m-1}$.
		If $o$ has no output neighbours, any $\corrset\sse O\setminus\{o\}$ is valid.
		If $o$ has output neighbours, fix one such neighbour $v\in N_G(o)\cap O$, then choose $\corrset'\sse O\setminus\{o,v\}$.
		Now if $o\in\odds{G}{\corrset'}$, set $\corrset := \corrset'\cup\{v\}$, and if $o\notin\odds{G}{\corrset'}$, set $\corrset := \corrset'$.
		In this way, every $\corrset'$ gives rise to exactly one valid choice of $\corrset$ and every valid $\corrset$ arises from exactly one $\corrset'$.
		Hence there are $\abs{\pow{O\setminus\{o,v\}}} = 2^{m-2}$ choices.
		
	\item Given $o$ and $\corrset$, if $\corrset=\emptyset$, there are $2^{m-1}$ choices for $\outneighb$.
		If $\corrset$ is non-empty, there are $2^{m-2}$ choices for $\outneighb$.
		The argument is similar to that for the previous point.
\end{itemize}
The overall number of options is still exponential, but it now depends on the fixed number of outputs rather than on the (increasing) total number of qubits in the computation.
Moreover, whether an insertion is possible can be deduced directly from the graph, there is no need to compute (or keep track of) a flow.

Consider now the different ways of generating some target computation on $n$ qubits, of which there are $n_I$ inputs and $n_O$ outputs.
Note that each sequence of output-splittings (i.e.\ applications of Proposition~\ref{prop:output-splitting}) that generates the target computation is associated with a totalisation of the gflow partial order over the measured vertices.
Therefore, the probability of generating a specific target computation (among all computations with the same numbers of inputs, of outputs, and of qubits) is proportional to the number of such totalisations.
In other words, a target computation which has a gflow with a `less restrictive' induced partial order (or even, if $n_O>n_I$, different gflows with different partial orders) will be more likely than a computation that has a `more restrictive' induced partial order.

\begin{example}
	For example, consider the following two diagrams with $n_I=n_O=2$ and $n=6$ that arise from the same initial input-output unfusion step shown in the middle:
	\[
		\input{generation-ex-3-3.tikz} \qquad\leftarrow\qquad
		\input{generation-ex-2-2.tikz} \qquad\rightarrow\qquad
		\input{generation-ex-4-2.tikz}
	\]
	For the left-hand side diagram, the partial order is $i_1\prec c\prec o_1$ and $i_2\prec d \prec o_2$, so there are two ways of generating the diagram from the one in the middle, depending on whether $c$ or $d$ is inserted first.
	For the right-hand side diagram, the partial order is $i_1\prec c\prec d\prec o_1$ and $i_2 \prec o_2$, so there is only one way of generating the diagram from the one in the middle.
	Thus the left-hand side diagram is twice as likely to be generated than the right-hand side one, assuming all choices made during the generation process are uniformly random.
\end{example}

\section{Conclusions}
\label{s:conclusions}

We have introduced a set of three rewrite rules for computations in the one-way model which preserve the existence of flow but not necessarily the interpretation.
The rule set is simple and minimal.
It can be used for Pauli flow and for extended gflow, where measurements can be in all three planes; both are preserved by the rules of Figure~\ref{fig:flow-preserving} (as long as \eqref{ZL} is only used to insert planar measurements in the gflow case).
Moreover, it is also possible to work in the setting where all measurements are $XY$: while the property of all measurements being labelled $XY$ is not preserved by individual applications of \eqref{LC} and \eqref{ZL}, these rules can nevertheless be composed into derived rewrite rules that do preserve the all-$XY$ property, e.g.\ \eqref{PD} and the lemmas in Section~\ref{s:derived}.

Regarding the conditions for flow preservation: the `deletion' direction of \eqref{ZL} and both directions of \eqref{IO} preserve flow unconditionally; \eqref{LC} breaks flow only if the top vertex (the centre of the local complementation) is an input, which is straightforward to avoid.
This leaves the `insertion' direction of \eqref{ZL} as the only complicated rule.
Nevertheless, as pointed out in Ref.~\cite{backens_inserting_2025}, at least when working with unitary one-way computations where $\abs{I}=\abs{O}$, it is more efficient to keep track of and update the unique focused flow, rather than run a full flow-finding algorithm for each insertion.
Additionally, following the approach of Section~\ref{s:generating} allows the generation of arbitrary all-$XY$ diagrams with gflow without needing to know the gflow at all.
We provide some basic analysis of the properties of this approach, such as what choices the generating algorithm needs to make and which diagrams are more or less likely to be generated.

It is interesting that our not necessarily interpretation preserving rule set is so close to the complete set of rewrite rules for one-way computations with only Pauli measurements \cite{mcelvanney_complete_2023}.
Indeed, the additional rules of the full flow- and interpretation-preserving rule set are all concerned with handling phase labels~\cite{backens_completeness_2026}, indicating a clear split between `structural' rules that belong to the Clifford fragment and `phase-managing' rules.
This may also be related to the special rule played by the traditional Clifford-fragment ZX-calculus rewrite rules in the context of ZX-flow, a version of which is unconditionally preserved by Clifford-fragment rewrites \cite{kissingerZXFlowFlexibleCriterion2026}.

While we have shown which rules are fundamentally necessary and sufficient for generating arbitrary one-way computations with flow, their practical application has not yet been explored.
Even from a more theoretical side, questions remain: for example, what strategies should be used to create interesting and useful computations with flow?
Or, how should one go about generating computations with flow that satisfy certain properties of interest?
We leave these questions to future work.

\subsection*{Acknowledgements}

Thank you to Ivica Turkalj for many interesting and useful conversations during the preparation of this paper.
I would also like to thank Piotr Mitosek for helpful comments about the completeness proof approach.

This work is supported by the Plan France 2030 through the PEPR integrated project EPiQ ANR-22-PETQ-0007 and the HQI platform ANR-22-PNCQ-0002; and by the European project MSCA Staff Exchanges Qcomical HORIZON-MSCA2023-SE-01.
The project is also supported by the Maison du Quantique MaQuEst.

\bibliographystyle{eptcs}
\bibliography{biblio,MBQC}

\end{document}